\makeatletter

\def\eqalign#1{\,\vcenter{\openup\jot\m@th
  \ialign{\strut\hfil$\displaystyle{##}$&$\displaystyle{{}##}$\hfil
      \crcr#1\crcr}}\,}
\def\eqalignno#1{\displ@y \tabskip\@centering
  \halign to\displaywidth{\hfil$\displaystyle{##}$\tabskip\z@skip
    &$\displaystyle{{}##}$\hfil\tabskip\@centering
    &\llap{$##$}\tabskip\z@skip\crcr
    #1\crcr}}
\def\leqalignno#1{\displ@y \tabskip\@centering
  \halign to\displaywidth{\hfil$\displaystyle{##}$\tabskip\z@skip
    &$\displaystyle{{}##}$\hfil\tabskip\@centering
    &\kern-\displaywidth\rlap{$##$}\tabskip\displaywidth\crcr
    #1\crcr}}

\makeatother

\newdimen\pixel \pixel=.00333333 true in

\documentclass[11pt]{article}
\usepackage{amsmath,amssymb,amsthm}
\usepackage{fullpage}
\usepackage{graphicx}

\usepackage[bookmarks,colorlinks,breaklinks]{hyperref}  
\hypersetup{linkcolor=blue,citecolor=blue,filecolor=blue,urlcolor=blue} 
\newcommand{\lref}[2][]{\hyperref[#2]{#1~\ref*{#2}}}
\renewcommand{\eqref}[1]{\hyperref[#1]{(\ref*{#1})}}
\numberwithin{equation}{section}
\usepackage[small]{titlesec}
\usepackage[small]{caption}

\def\begex{\begin{example}\parindent=0pt \rm}
\def\endex{\end{example}}
\def\square{\vbox{\hrule height.2pt\hbox{\vrule width.2pt height5pt \kern5pt
                                   \vrule width.2pt} \hrule height.2pt}}

\def\footnoterule{\kern-.5\pixel
        \hrule height \pixel width \columnwidth
        \kern-.5\pixel}

\addtocounter{footnote}{1}

\newtheorem{theorem}{Theorem}[section]
\newtheorem{lemma}[theorem]{Lemma}

\newtheorem{corollary}[theorem]{Corollary}

\theoremstyle{definition}

\newtheorem{example}[theorem]{Example}

\theoremstyle{remark}
\newtheorem{remark}[theorem]{Remark}

\makeatletter
\def\bigpar{\bigbreak\@afterindentfalse\@afterheading\ignorespaces}
\def\medpar{\medbreak\@afterindentfalse\@afterheading\ignorespaces}
\def\smallpar{\smallbreak\@afterindentfalse\@afterheading\ignorespaces}

\newcommand{\field}{\mathbb{F}}

\newcommand{\naturals}{\mathbb{N}}

\newcommand{\sgn}{\mathrm{sgn}}
\newcommand{\con}{\mathrm{con}}
\newcommand{\CC}{\mathcal{C}}
\newcommand{\per}{\mathrm{per}}

\newcommand{\XOR}{\mathrm{XOR}}
\newcommand{\poly}{\mathrm{poly}}
\newcommand{\F}{\mathbb{F}}
\newcommand{\numP}{\mbox{\rm \#P}}

\newcommand{\setcond}[2]{\left\{#1\: \middle|\: #2\right\}}

\begin{document}
\title{Almost Settling the Hardness of Noncommutative Determinant}

\author{
Steve Chien\thanks{Microsoft Research, Silicon Valley,  1065 La Avenida, Mountain View CA 94043, USA. email: {\tt schien@microsoft.com}.}
\and
Prahladh Harsha\thanks{Tata Institute of Fundamental Research, Homi Bhabha Road, Mumbai 400005, INDIA. email: {\tt prahladh@tifr.res.in}. Part of this work was done while the author was at Microsoft Research, Silicon Valley and the MIT Computer Science and Artificial Intelligence Laboratory.}
\and
Alistair Sinclair\thanks{Computer Science Division, University of California
Berkeley, CA 94720, USA.
email: {\tt sinclair@cs.berkeley.edu}.}
\and
Srikanth Srinivasan\thanks{Institute for Advanced Study, Einstein Drive, Princeton NJ 08540, USA.
email: {\tt srikanth@math.ias.edu}. Part of this work was done while the author was at Microsoft Research, Silicon Valley.}}

\maketitle
\thispagestyle{empty}

\begin{center}
\begin{abstract}
In this paper, we study the complexity of computing the determinant of a matrix over a non-commutative algebra. In particular, we ask the question, ``over which algebras, is the determinant easier to compute than the permanent?'' Towards resolving this question, we show the following {\em hardness} and {\em easiness} of noncommutative determinant computation.
\begin{itemize}
\item {[Hardness]} Computing the determinant of an $n\times n$ matrix whose entries are themselves $2\times 2$ matrices over a field is as hard as computing the permanent over the field. This extends the recent result of Arvind and Srinivasan, who proved a similar result which however required the entries to be of linear dimension.
\item {[Easiness]} Determinant of an $n\times n$ matrix whose entries are themselves $d\times d$ upper triangular matrices can be computed in $\poly(n^d)$ time.
\end{itemize}
Combining the above with the decomposition theorem of finite dimensional algebras (in particular exploiting the simple structure of $2\times 2$ matrix algebras), we can extend the above hardness and easiness statements to more general algebras as follows. Let $A$ be a finite dimensional algebra over a finite field with radical $R(A)$.
\begin{itemize}
\item {[Hardness]} If the quotient $A/R(A)$ is non-commutative, then computing the determinant over the algebra $A$ is as hard as computing the permanent.
\item {[Easiness]} If the quotient $A/R(A)$ is commutative and furthermore, $R(A)$ has nilpotency index $d$ (i.e., the smallest $d$ such that $R(A)^d =0$), then there exists a $\poly(n^d)$-time algorithm that computes determinants over the algebra $A$.
\end{itemize}
In particular, for any constant dimensional algebra $A$ over a finite field, since the nilpotency index of $R(A)$ is at most a constant, we have the following dichotomy theorem: if $A/R(A)$ is commutative, then efficient determinant computation is feasible and otherwise determinant is as hard as permanent.
\end{abstract}
\end{center}
\newpage

\setcounter{footnote}{1}
\setcounter{page}{1}
\setcounter{section}{0}
\section{Introduction}\label{sec:intro}

Given a matrix $M=\{m_{ij}\}$, the determinant of $M$, denoted by $\det(M)$ is given by the polynomial $\det(M) = \sum_{\sigma \in S_n}\sgn(\sigma)\prod_{i = 1}^n m_{i\sigma i}$, while the permanent of $M$, denoted by $\per(M)$ is defined by the polynomial $\per(M) = \sum_{\sigma \in S_n}\prod_{i = 1}^n m_{i\sigma i}.$ Though deceivingly similar in their definitions, the determinant and permanent behave very differently with respect to how efficiently one can compute these quantities. The determinant of a matrix over any field can be efficiently computed using Gaussian elimination. In fact, determinant continues to be easy even when the entries come from some commutative algebra, not necessarily a field~\cite{Samuelson1942,Berkowitz1984,Chistov1985,MahajanV1997}. Computing the permanent of a matrix over the rationals, on the other hand, as famously shown by Valiant~\cite{Valiant1979}, is just as hard as counting the number of satisfying assignments to a Boolean formula or equivalently \numP-complete even when the entries are just 0 and 1. Given this state of affairs, it is natural to ask, ``what is it that makes the permanent hard while the determinant is easy?'' Understanding this distinction in complexity of computing the determinant and permanent of a matrix is a fundamental problem in theoretical computer science.

Nisan first pioneered the study of noncommutative lower bounds in his 1991 groundbreaking paper~\cite{Nisan1991}. In one of that paper's more important results, Nisan proves that any algebraic branching program (ABP) that computes the determinant of a matrix $M = \{m_{ij}\}$ over the noncommutative free algebra $\F\langle x_{11},\ldots,x_{nn}\rangle$ must have exponential size; this then implies a similar lower bound for arithmetic formulas. This contrasts markedly with the many known efficient algorithms for determinant in commutative settings, which include polynomial-sized ABPs~\cite{MahajanV1997}. 

This problem takes on added significance in light of a connection discovered by Godsil and Gutman~\cite{GodsilG1981} and developed by Karmarkar et al.~\cite{KarmarkarKLLL1993} between computing determinants and exponential time algorithms for approximating the permanent. The promise of this approach was cemented when Chien et al.~\cite{ChienRS2003}, expanding on work by Barvinok~\cite{Barvinok1999}, showed that if one can efficiently compute determinant of an $n \times n$ matrix $M$ whose entries $m_{ij}$ are themselves matrices of $O(n^2)$ dimension, then there is a fully polynomial randomized approximation scheme for the permanent of a 0-1 matrix; similar results were later proven by Moore and Russell~\cite{MooreR2009}. Thus understanding the complexity of noncommutative determinant is of both algorithmic and complexity-theoretic importance.

Nisan's results are somewhat limited in that they apply only to the free algebra $\F\langle x_{i}\rangle$ and not to specific finite dimensional algebras (such as those used to approximate the permanent), and because they do not apply outside of ABPs and arithmetic formulas. Addressing the first concern, Chien and Sinclair~\cite{ChienS2007} significantly strengthened Nisan's original lower bounds to apply to a wide range of other algebras by analyzing those algebras' polynomial identities. In particular, they show that Nisan's lower bound extends to $d \times d$ upper-triangular matrix algebra over a field of characteristic $0$ for any $d > 1$ (and hence over $M_d(\F)$, the full $d\times d$ matrix algebra as well), the quaternion algebra, and several others, albeit only for ABPs.

In a significant advance, Arvind and Srinivasan~\cite{ArvindS2010} recently broke the ABP barrier and showed noncommutative determinant lower bounds for much stronger models of computation. They show that unless there exist small circuits to compute the permanent, there cannot exist small noncommutative circuits for the noncommutative determinant. More devastatingly from the algorithmic point of view, they show that computing $\det(M)$ where the $m_{ij}$ are linear-sized matrix algebras is at least as hard as (exactly) computing the permanent. Arvind and Srinivasan thus bring into serious doubt whether the determinant-based approaches to approximating the permanent are computationally feasible.

While these collections of results make substantial progress in our understanding of when determinant can be computed over a noncommutative algebra, they are still incomplete in significant ways. First, we do not know whether Arvind and Srinivasan's results rule out algorithms for determinants over constant-dimensional matrix algebras, which are still of use in approximating the permanent. More expansively, we still do not know the answer to what is perhaps the fundamental philosophical question underlying this:

\begin{quote} {\em Whether there is any noncommutative algebra over which we can compute determinants efficiently, or whether, as may seem attractive, commutativity is a necessary condition to having such algorithms?}
\end{quote}

\subsection{Our results}
In this paper, we fill in most of these remaining gaps. Our first main result extends Arvind and Srinivasan's results all the way down to $2\times 2$ matrix algebras.
\begin{theorem}\label{thm:2x2-intro} (stated informally\footnote{See \lref[Theorem]{thm:2x2} for a formal statement.})
Let $M_2(\F)$ be the algebra of $2\times 2$ matrices over a field $\F$. Then computing the determinant over $M_2(\F)$ is as hard as computing the permanent over $\F$.
\end{theorem}
The proof of this theorem works by retooling Valiant's original reduction from \#3SAT to permanent. One would not expect to be able to modify Valiant's reduction to go from \#3SAT to determinant over a field $\F$, as there are known polynomial-time algorithms in that setting. However, when working with $M_2(\F)$, what we show is that there is just enough noncommutative behavior in $M_2(\F)$ to make Valiant's reduction (or a slight modification of it) go through.

Given the central role of matrix algebras in ring theory, this allows us to prove similar results for other large classes of algebras. In particular, consider a finite-dimensional algebra $A$ over a finite field $\F$. This algebra has a {\em radical} $R(A)$, which happens to be a nilpotent ideal of $A$. Combined with classical results from algebra (in particular the simple structure of the $2\times 2$ matrix algebras) the above theorem can be extended as follows to yield our second main result.
\begin{theorem}\label{thm:radical-intro} (stated
	informally\footnote{See \lref[Theorem]{thm_algebra_hardness} for a formal statement.})
	If $A$ is a fixed\footnote{By \emph{fixed}, we mean that the
	algebra is not part of the input; we fix an algebra $A$ and consider
	the problem of computing the determinant over $A$.} finite dimensional algebra over a finite field such that the quotient $A/R(A)$ is noncommutative, then computing determinant over $A$ is as hard as computing the permanent.
\end{theorem}
In particular, if the algebra is semisimple (i.e, $R(A) = 0$), then
the commutativity of $A$ itself is determinative: if $A$ is
commutative, there is an efficient algorithm for computing $\det$ over $A$; otherwise, it is at least as hard as computing the permanent. The class of semisimple algebras includes several well-known examples, such as group algebras.

It may be tempting at this point to see the sequence of lower bounds
starting from Nisan's original work and conjecture that computing
$\det$ over $A$ for some algebra $A$ is feasible if and only if $A$ is
commutative. Perhaps surprisingly, we show that this is not the
case---in, fact there do exist noncommutative algebras $A$ for
which there are polynomial-time algorithms for computing $\det$ over $A$. For instance, in our third main result, we show that computing the determinant where the matrix entries are $d\times d$ upper triangular matrices for constant $d$ is easy. For reasons that will soon be clear, we will state this result, more generally, in the language of radicals.
\begin{theorem}\label{thm:alg-intro}
Given a finite dimensional algebra $A$ and its radical $R(A)$, let $d$ be the smallest value for which $R(A)^d = 0$ (i.e.~any product of $d$ elements of $R(A)$ is $0$). If $A/R(A)$ is commutative, there is an algorithm for computing $\det$ over $A$ in time $\poly(n^{d})$.
\end{theorem}
While this description of the class of algebras that allow efficient determinant computation is somewhat abstruse, it does include several familiar algebras. Perhaps most familiar is the algebra $U_d(\F)$ of $d\times d$ upper-triangular matrices, for which $R(U_d(\F))^d = 0$. What the result states is that the key to whether determinant is computationally feasible is not commutativity alone. For noncommutative algebras, it is still possible that determinant can be efficiently computed, so long as all of the noncommutative elements belong to a nilpotent ideal and have a limited ``lifespan'' of sorts.

The above theorems together yield a nice dichotomy for constant dimensional algebras over a finite field. Given any such algebra $A$ of constant dimension $D$ over a finite field, either $A/R(A)$ is commutative or not. Furthermore, if $A/R(A)$ is commutative, we have that $R(A)$ is nilpotent with nilpotency index at most $D$ which is a constant. We thus, have the following dichotomy: if $A/R(A)$ is commutative, then efficient determinant is feasible else determinant is as hard as permanent.

Does this yield a complete characterization of algebras over which efficient determination computation is feasible? Unfortunately not. In particular, what if the dimension $D$ is  non-constant, i.e., the algebra is not fixed but given as part of the input or if the algebra is over a field of characteristic 0? In these cases, the lower bound of \lref[Theorem]{thm:radical-intro} and upper bound of \lref[Theorem]{thm:alg-intro} are arguably close, but do not match. A complete characterization remains an intriguing open problem.

\paragraph{Organization of the paper:} After some preliminaries in
\lref[Section]{sec:prelim}, we prove lower and upper bounds in two
concrete settings: we prove a lower bound for $2\times 2$ matrix
algebras in \lref[Section]{sec:lower} and an upper bound for
small-dimensional upper triangular matrix algebras in
\lref[Section]{sec:upper}. The results on general algebras are in
\lref[Section]{sec:general}, followed by some discussion in
\lref[Section]{sec:disc}.

\section{Preliminaries}\label{sec:prelim}
In this section we define terms and notation that will be useful later.

An (associative) algebra $A$ over a field $\field$ is a vector
space over $\field$ with a bilinear, associative multiplication
operator that distributes over addition. That is, we have a map
$\cdot: A\times A\rightarrow A$ that satisfies: (a) $x\cdot(y\cdot z)
= (x\cdot y)\cdot z$ for any $x,y,z\in A$, (b) $\lambda (x\cdot y) =
(\lambda x)\cdot y = x\cdot (\lambda y)$, for any $\lambda\in\field$
and $x,y\in A$, and (c) $x\cdot (y+z) = x\cdot y + x\cdot z$ and
$(y+z)\cdot x = y\cdot x + z\cdot x$ for any $x,y,z\in A$. We will
assume that all our algebras are \emph{unital}, i.e., they contain an
identity element. We will denote this element as $1$.  For more about
algebras, see Curtis and Reiner's book \cite{CurtisR1962}.  A tremendous range
of familiar objects are algebras; we will be concerned with the
algebra of $d \times d$ matrices over $\F$, which we will denote
$M_d(\F)$, as well as the algebra of $d \times d$ upper-triangular
matrices over $\F$, or $U_d(\F)$. Other prominent examples are the
free algebra $\F\langle x_i \rangle$, the algebra of polynomials
$\F[x_i]$, group algebras over a field, or a field considered as an
algebra over itself.

Given an $n \times n$ matrix $M = (m_{ij})$ whose elements belong to an algebra $A$, the {\em determinant} of $M$, or $\det(M)$, is defined as the polynomial
$\det(M) = \sum_{\sigma \in S_n}\sgn(\sigma)\prod_{i = 1}^n m_{i\sigma i}.$
Note that when $A$ is noncommutative, the order of the multiplication becomes important. When the order is by row, as above, we are working with the Cayley determinant. The permanent of the same matrix is
$\per(M) = \sum_{\sigma \in S_n}\prod_{i = 1}^n m_{i\sigma i}.$  We will denote by $\det_A$ (and $\per_A$) the problem of computing the determinant (and permanent) over an algebra $A$.

We recall also the familiar recasting of the determinant and permanent in terms of cycle covers on a graph.
Suppose $M = (m_{ij})$ is an $n\times n$ matrix over an algebra $A$.  Let $G(M)$ denote the weighted directed graph on vertices $1,\ldots,n$ that has $M$ as its adjacency matrix. A permutation
$\pi:[n]\rightarrow [n]$ from the rows to the columns of $M$ can be identified with the set of edges $(i,\pi(i))$ in the
graph $G(M)$; it is easily observed that these edges form a (directed) cycle cover of $G(M)$. Letting $\CC(G)$ denote the collection of all cycle covers of $G(M)$, we can write

\begin{equation}\label{eqn:det}
\det(M) = \sum_{C\in\CC(G(M))}
\sgn(C)m_{1,C(1)}m_{2,C(2)}\cdots m_{n,C(n)}
\end{equation}
and
\begin{equation}\label{eqn:per}
\per(M) = \sum_{C\in\CC(G(M))}
m_{1,C(1)}m_{2,C(2)}\cdots m_{n,C(n)},
\end{equation}
where for a given cycle cover $C$, $C(i)$ represents the successor of vertex $i$ in $C$, and $\sgn(C)$ is the sign of $C$. It is known that $\sgn(C) = (-1)^{n - c}$, with $c$ being the number of cycles in $C$, and that this is also the sign of the corresponding permutation. We will denote the weight of an edge $e =(x,y)$ as $w(e)$ or $w(x,y)$. Further, for a subset of edges $B = \{(x_1,y_1),\ldots,(x_{|B|},y_{|B|})\}$ of a cycle cover $C$ with $x_i < x_{i+1}$, we can define the weight of $B$ as $w(B) = \prod_{i = 1}^{|B|} w(x_i,y_i)$. (Note that the product is in order by source vertex.) Thus $w(C) = \prod_i m_{i,C(i)}$ by $w(C)$ is the weight of the cycle cover, and the product $\sgn(C)\prod_i m_{i,C(i)}$ is the {\em signed weight} of $C$.

\section{The lower bound for $2\times 2$ matrix algebras}\label{sec:lower}
In this section, we show our key lower bound for $2 \times 2$ matrix algebras. Our proof is based on Valiant's seminal reduction from \#3SAT to permanent, as modified by Papadimitriou~\cite{Papadimitriou-complexity} and also described in the complexity textbook by Arora and Barak~\cite{AroraB2009}. We first give a self-contained description of that, before detailing our modifications of it.

\subsection{Valiant's lower bound for the permanent}\label{sec:valiant}
Valiant's reduction is from \#3SAT to permanent; given a \#3SAT formula $\varphi$ on $n$ variables and $m$ clauses, he constructs
a weighted directed graph $G_\varphi$ on $\poly(n,m)$ vertices such that the number of satisfying assignments of $\varphi$ is equal to ${\rm constant} \times \per(M(G_\varphi))$, where $M(G_\varphi)$ is the adjacency matrix of $G_\varphi$. The key components of $G_\varphi$ are the variable, clause, and XOR gadgets shown in \lref[Figure]{fig:gadgets}.\footnote{We follow a convention from~\cite{AroraB2009} in allowing gadgets to sometimes have multiple edges between the same two vertices. While technically prohibited in a graph defined by a matrix, this can be fixed by adding an extra node in these edges.} The idea is that there will be a relation between satisfying assignments of $\varphi$ and cycle covers of $G_\varphi$; moreover, for each satisfying assignment, the total weight of its corresponding cycle covers will be the same.

\begin{figure}
\begin{center}
\includegraphics[width=5.0in]{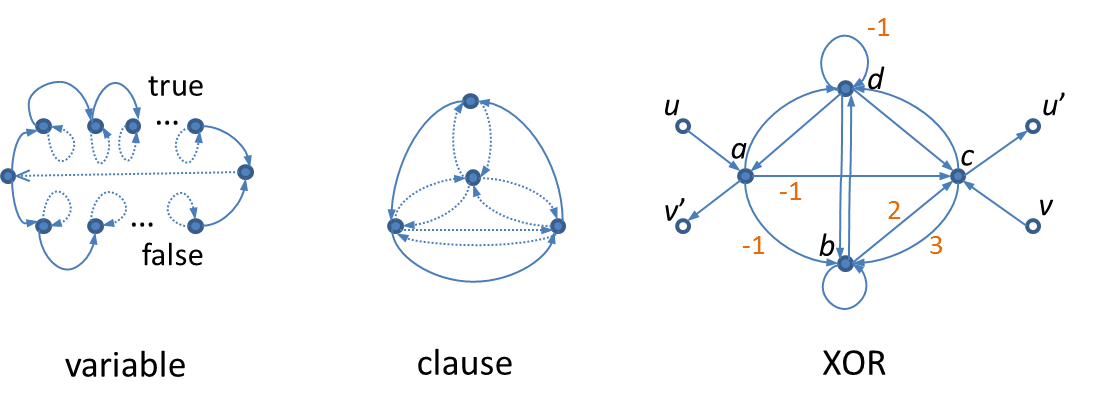}\\
\caption{Gadgets used in proving Valiant's lower bound. All edges have weight 1 unless noted otherwise. For the variable and clause gadgets, the solid (not dotted) edges are called {\em (vertex) external edges} or {\em (clause) external edges}. Note that the number of external edges in a variable gadget is not fixed, and need not be the same for the True and False halves of the gadget.}\label{fig:gadgets}
\end{center}
\end{figure}

Before defining $G_\varphi$ itself, we first work with a preliminary graph $G^0_\varphi$ that contains $n$ variable gadgets and $m$ clause gadgets, but no XOR gadgets; all of the gadgets are disjoint from each other. For the moment, the number of external edges in each of the variable gadgets is unimportant. In analyzing $G^0_\varphi$, we will use the following:
\begin{lemma}\label{lem:var-clause}
The following hold for the gadgets in \lref[Figure]{fig:gadgets}:
(a) A variable gadget has exactly two cycle covers. Each cycle cover contains one long cycle using all of the external edges on one side of the gadget and the long middle edge, as well as all the self-loops on the other side of the gadget.
(b) In a clause gadget, there is no cycle cover that uses all three
		external edges. For every \emph{proper} subset $S$ of the external edges in a clause gadget,
		there is exactly one cycle cover that contains exactly the edges
		in $S$; this cycle cover has weight $1$.
\end{lemma}
As all $n + m$ gadgets in $G^0_\varphi$ are disjoint, any cycle cover of $G^0_\varphi$ will be a union of $n + m$ smaller cycle covers--namely, one for each gadget. The choice of cycle cover for each gadget defines the value of each variable and which literals are satisfied in each clause.

More precisely, for a variable gadget, let the term {\em True cycle cover} denote the cycle cover containing the external edges on the True side of the gadget. Analogously, the {\em False cycle cover} refers to the cycle cover containing the external edges on the False side of the gadget. The idea is that a cycle cover of $G^0_\varphi$ sets a variable to T or F by choosing either the True or False cycle cover. Meanwhile, for clause gadgets, the intention is that each external edge will correspond to one of the three literals in the clause, and an external edge is used in a cycle cover if and only if the corresponding literal is set to F (i.e.~the corresponding literal is {\em not} satisfied). Since no cycle cover can contain all three external edges of a clause gadget, in this interpretation at least one of the literals in the clause must be satisfied.

We say a cycle cover $C$ of $G^0_\varphi$ is {\em consistent} if (1) whenever $C$ contains the True cycle cover of the gadget for a variable $x_k$, it contains {\em all} clause external edges for instances of the negative literal $\overline{x_k}$ and {\em no} clause external edges for instances of the positive literal $x_k$, and (2) conversely, whenever $C$ contains the False cycle cover for $x_k$, it contains all clause external edges for instances of $x_k$ but no clause external edges for instances of $\overline{x_k}$. A consistent cycle cover therefore does not ``cheat'' by claiming to set $x_k$ to T (for example) in a variable gadget but to F in a clause gadget. This is close to what we want:

\begin{lemma}\label{lem:consistent}
The number of satisfying assignments of $\varphi$ is equal to the total weight of consistent cycle covers of $G^0_\varphi$.
\end{lemma}
\begin{proof}
This follows from combining the natural bijection between satisfying assignments and consistent cycle covers and the fact from \lref[Lemma]{lem:var-clause} that every cycle cover of a clause gadget has weight $1$.
\end{proof}

Of course, nothing about $G^0_\varphi$ guarantees that a cycle cover must be consistent, and in fact many inconsistent covers exist.
To fix this, we need to use the critical XOR gadgets to obtain the final graph $G_\varphi$.

The graph $G_\varphi$ is constructed as shown in \lref[Figure]{fig:combined-graph} (left). It has the same $n$ variable gadgets and $m$ clause gadgets as $G^0_\varphi$, with the gadget for each variable $x_k$ having as many True external edges as there are instances of $x_k$ in $\varphi$, and as many False external edges as there are instances of $\overline{x_k}$. Now, however, for each appearance of a literal $x_k$ or $\overline{x_k}$ in a given clause, an XOR gadget is used to replace the corresponding external edge in that clause gadget and a distinct external edge on the appropriate side of the variable gadget for $x_k$. The role of the XOR gadgets is to neutralize the inconsistent cycle covers of $G^0_\varphi$ while still maintaining the property that each satisfying assignment of $\varphi$ contributes the same to the total weight of cycle covers. This leads to the description of the final graph $G_\varphi$ itself.

We now state the important properties of the XOR gadget, the key component of Valiant's proof.
\begin{lemma}\label{lem:XOR}
Suppose a graph $G$ contains edges $(u,u')$ and $(v,v')$, with all four vertices distinct.
Suppose now that the edges $(u,u')$ and $(v,v')$ are replaced by an XOR gadget as shown in \lref[Figure]{fig:gadgets}, resulting in a new graph $G'$ (with four new vertices $a,b,c$ and $d$).
Let $\CC_{u \setminus v}$ be the set of cycle covers containing $(u,u')$ but not $(v,v')$, and $w_{u\setminus v} = \sum_{C \in \CC_{u \setminus v}}w(C)$ be their total weight. Let $\CC_{v \setminus u}$ and $w_v = \sum_{C \in \CC_{v \setminus u}}w(C)$ be defined analogously.
Then there exist two disjoint sets of cycle covers of $G'$ with total weight $4w_{u\setminus v}$ and $4w_{v\setminus u}$, while all cycle covers of $G'$ not in these sets have total weight $0$.
\end{lemma}
The proof is omitted, as we will state and prove our own modified version of this in \lref[Section]{sec:2x2}.

\begin{figure}
  \begin{center}
  \includegraphics[width=5.9in]{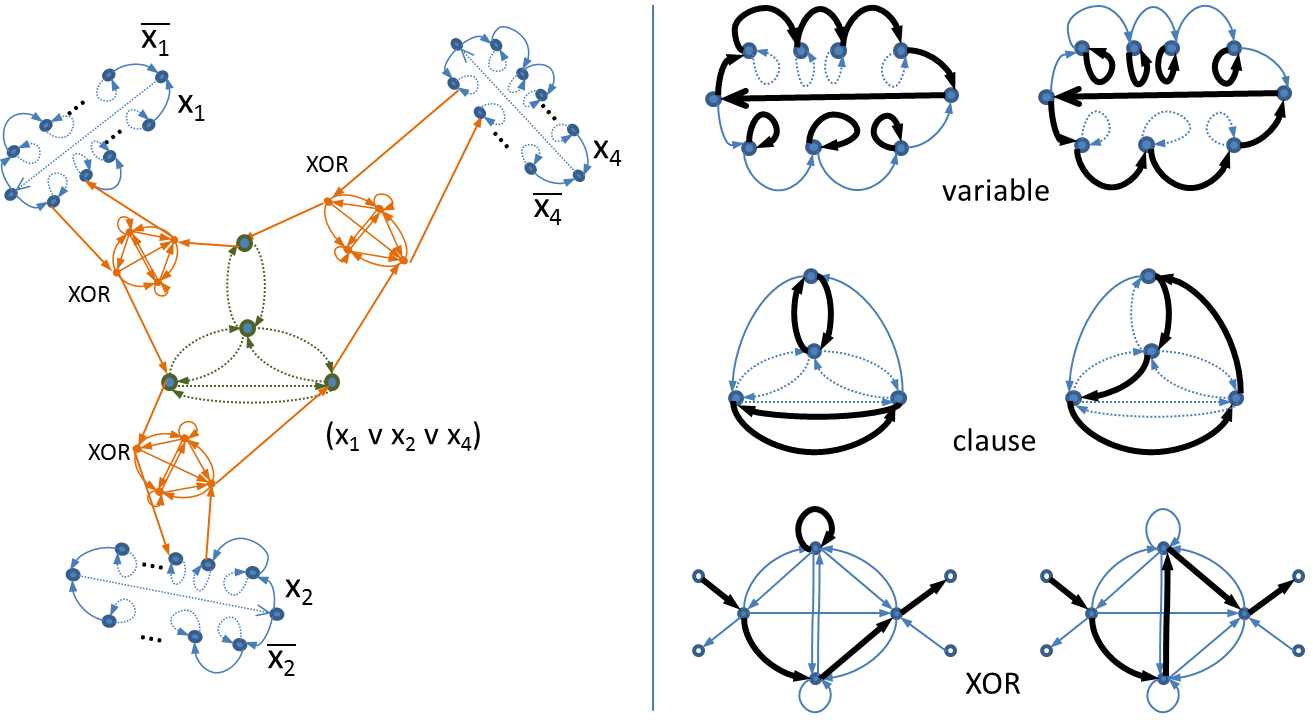}\\
  \caption{{\bf Left:} Subgraph of $G_\varphi$ corresponding to clause $(x_1 \vee x_2 \vee \overline{x_4})$, with the clause gadget in center. Three variable gadgets are connected to the clause gadget via XOR gadgets. {\bf Right:} Examples of how gadgets may have cycle covers of different sign.}\label{fig:combined-graph}
  \end{center}
\end{figure}
This leads to the following:
\begin{theorem}\label{thm:valiant}{\bf [Valiant]}
Given a 3-SAT formula $\varphi$ and the graph $G_\varphi$ as described, $\per(G_\varphi) = 4^{3m}S$, where $S$ is the number of satisfying assignments of $\varphi$.
\end{theorem}
We omit the formal proof, but give some of the intuition. Beginning with $G^0_\varphi$, we begin adding XOR gadgets one at a time.
When a pair of edges is replaced by an XOR gadget, any cycle covers that are consistent with respect to that pair of edges are turned into a set of cycle covers whose total weight is a factor of $4$ more than the original weight. All other cycle covers in the new graph have total weight $0$. This continues until each of the $3m$ XOR gadgets are added, at which point the original consistent cycle covers have become a set of cycle covers with total weight $4^{3m}$ while all other cycle covers in the final graph have weight $0$. The total weight of the cycle covers in the final graph is therefore $4^{3m}S$, as required.

\subsection{Our construction}\label{sec:2x2}
We now prove the following:
\begin{theorem}\label{thm:2x2}
Let $\F$ be a field of characteristic $p\geq 0$. If $p=0$,
computing $\det_{M_2(\F)}$ is $\numP$-hard. On the other hand, if
$p>0$ and odd, then computing $\det_{M_2(\F)}$ is $\mathrm{Mod}_pP$-hard.
\end{theorem}

Our proof is also a reduction from \#3SAT (or $\mathrm{Mod}_p$-SAT in
the case of positive odd characteristic) and is based on Valiant's framework as described in the previous subsection. Given a 3SAT formula $\varphi$, we wish to construct a directed graph $H_\varphi$ with weights belonging to $M_2(\F)$ such that the number of satisfying assignments of $\varphi$ can be computed from $\det(M(H_\varphi))$, as expressed in equation~\eqref{eqn:det} above. We will first describe the graph and then prove its correctness.

A very naive but instructive first try would be to simply use the graph $G_\varphi$ from Valiant's construction, replacing each edge weight $w \in \F$ with $wI$, where $I_2$ is the $2\times 2$ identity matrix. This fails, of course, because of the sign factor $\sgn(C)$ inside the summation, which is based on the parity of the number of cycles in $C$. The immediate problem is that each of the three types of gadgets could conceivably use an odd or even number of cycles. As shown in \lref[Figure]{fig:combined-graph} (right), variable gadgets may have a different number of self-loops on different sides; clause gadgets may use one or two cycles depending on which external edges are chosen; and XOR gadgets show similar behavior.

Fortunately, these problems can be overcome if we also allow ourselves to modify the edge weights, and crucially, use the noncommutative structure available in $M_2(\F)$. This results in the gadgets shown in \lref[Figure]{fig:new-gadgets}. We now define two graphs, a preliminary graph $H^0_\varphi$ and final graph $H_\varphi$, in analogy with $G^0_\varphi$ and $G_\varphi$ from \lref[Section]{sec:valiant}. The new graphs $H^0_\varphi$ and $H_\varphi$ will be constructed in the same manner as $G_\varphi$, only using the modified gadgets from \lref[Figure]{fig:new-gadgets} instead of the original gadgets in \lref[Figure]{fig:gadgets}.


The rough idea behind these gadgets is that with the new weights, each resulting cycle cover of a gadget of the ``wrong'' sign will have an extra $-1$ sign from its edge weights. The determinant is then essentially the same as the permanent. We now explain the changes in more detail.

\begin{figure}
\begin{center}
  \includegraphics[width=5.0in]{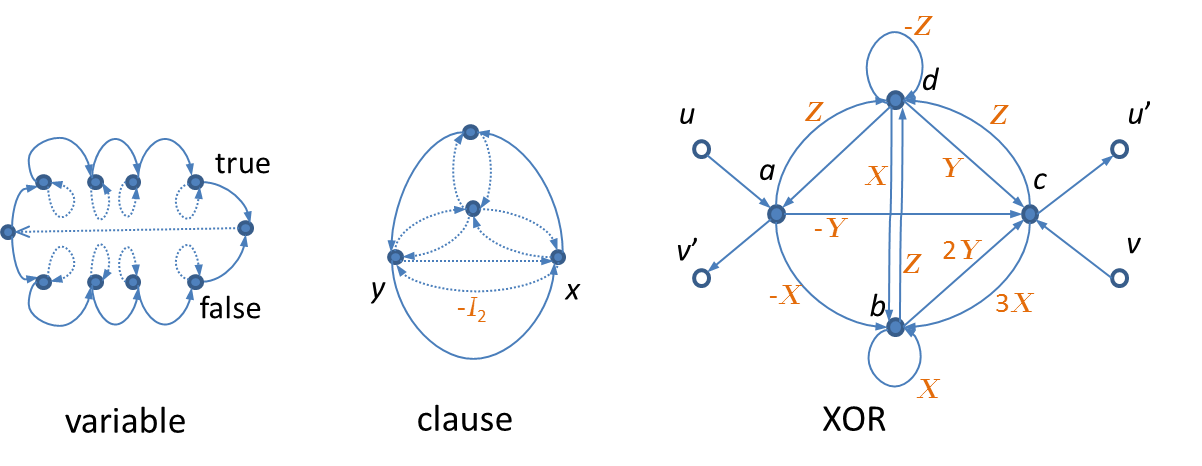}\\
  \caption{Modified gadgets.}\label{fig:new-gadgets}
\end{center}
\end{figure}

For variable gadgets, the fix is easy -- all we have to do is make sure that both sides of the gadget have an even (for example) number of vertices, and hence an even number of self loops. This can be accomplished by adding, if necessary, a new vertex and appropriate new edges on one or both sides. The new external edges, if any, will not be connected to any of the clause gadgets.

For clause gadgets, we need to address the problem that some cycle covers have only one cycle, while others have two. Here we benefit from the observation that one of the edges, $(x,y)$ in \lref[Figure]{fig:new-gadgets}, is used only in cycle covers with two cycles. Thus we can correct for parity by changing the sign of this edge from $I_2$ to $-I_2$; as a result, every cycle cover of a clause gadget has the same signed weight.

For XOR gadgets, simply changing the edge weights to scalar multiples of $I_2$ is insufficient. (Indeed, Valiant presciently noticed this in 1979!) However, we can save the construction by using more sophisticated matrix-valued edge weights instead. In particular, we define the following three $2\times 2$ matrices:
\begin{equation}\label{eqn:xyz}
X = \left(
       \begin{array}{cc}
         1 & 0 \\
         0 & -1 \\
       \end{array}
     \right);
Y = \left(
      \begin{array}{cc}
        0 & -1 \\
        -1 & 0 \\
      \end{array}
    \right);
Z = \left(
     \begin{array}{cc}
       0 & -1 \\
       1 & 0 \\
     \end{array}
   \right).
\end{equation}
We then modify the weights of the edges between vertices $a,b,c$ and $d$. Specifically, each edge entering vertex $b$ has its weight multiplied by $X$; each edge entering $c$ has its weight multiplied by $Y$, and each edge entering $d$ has its weight multiplied by $Z$.

For now, with $H_\varphi$ defined, we prove that computing $\det(H_\varphi)$ is equivalent to computing the number of satisfying assignments of $\varphi$. We first observe the following analogue of \lref[Lemma]{lem:consistent}.

\begin{lemma}\label{lem:new-consistent}
Let $\CC_\con$ be the set of all consistent cycle covers of $H^0_\varphi$. Then there exists $z \in \{1,-1\}$ such that for all $C \in \CC_\con$, we have $\sgn(C)w(C) = zI_2$.
\end{lemma}
\begin{proof}
As in the proof of \lref[Lemma]{lem:consistent}, there is a bijection between satisfying assignments of $\varphi$ and consistent cycle covers of $H^0_\varphi$. We need to show that each of these cycle covers has the same signed weight. For such a cycle cover $C \in \CC_\con$ we have $\sgn(C) = (-1)^{n^0_H - c(C)}$, where $n^0_H$ is the number of vertices in $H^0_\varphi$ and $c(C)$ is the number of cycles in $C$. We further know that $(-1)^{c(C)} = (-1)^{p + m + q}$, where $p$ is the number of cycles used to cover the $n$ variable gadgets, $m$ is the number of clauses, and $q$ is the number of times $C$ uses two cycles to cover a clause gadget. Since we assumed $p$ to be even, we have $\sgn(C) = (-1)^{n^0_H + m + q}$.

On the other hand, $w(C)$ is the product of the edge weights of $C$. All of these weights are $I_2$ except for the $w(x,y)$ in the clause gadget, which has weight $-I_2$ and shows up when $C$ uses two edges for a clause gadget. Thus $w(C) = (-1)^q I_2$, and $\sgn(C)w(C) = (-1)^{n^0_H + m}I_2$, which is independent of the cycle cover $C$. (Hence, $\sum_{C \in \CC_\con} \sgn(C)w(C) = (-1)^{n^0_H + m} SI_2$, where $S$ is the number of satisfying assignments of $\varphi$).
\end{proof}
Without loss of generality, we can assume from here on that the sign $z$ is positive, as we can insert a new vertex within an edge so that $n^0_H + m$ is even.

We now prove the following useful identities of XOR gadgets, which can be verified by hand:
\begin{lemma}\label{lem:dets}
Let $M_\XOR$ be the adjacency matrix for the XOR gadget, or
$$\left(
  \begin{array}{cccc}
    0 & -X & -Y & Z \\
    0 & X & 2Y & Z \\
    0 & 3X & 0 & Z \\
    I_2 & X & Y & -Z \\
  \end{array}
\right).$$ Letting $M_{i,j}$ indicate the minor of $M$ with row $i$ and column $j$ removed, we have (1) $\det(M_{3,1}) = -4I_2$, (2) $\det(M_{1,3}) = -4J_2$, (3) $\det(M) = \det(M_{1,1}) = \det(M_{3,3}) = \det(M_{13,13}) = 0$, where $J_2 = \left(
                                                                            \begin{array}{cc}
                                                                              0 & 1 \\
                                                                              1 & 0 \\
                                                                            \end{array}
                                                                          \right)$.
\end{lemma}

Now consider a graph $G$ with vertices labeled $1,\ldots,n_G$ and weights in $M_2(\F)$. Suppose $G$ contains vertex-disjoint edges $(u,u')$ and $(v,v')$, each with weight $I_2$. Suppose now that the edges $(u,u')$ and $(v,v')$ are replaced by an XOR gadget as shown in \lref[Figure]{fig:new-gadgets}. This results in a new graph $G'$, with four new vertices $a,b,c$ and $d$, which we number $n_G+1,\ldots,n_G+4$. We now define a mapping $\psi$ from $\CC(G)$ to subsets of $\CC'(G)$ as follows: Given cycle covers $C \in \CC(G)$ and $C' \in \CC(G')$, then $C' \in \psi(C)$ if and only if (1) for all edges $e \in C\setminus \{(u,u'),(v,v')\}$, we have $e \in C'$, (2) $(u,u') \in C$ if and only if $(u,a),(c,u') \in C'$, and (3) $(v,v') \in C$ if and only if $(v,c),(a,v') \in C'$.

This leads to the following analogue of \lref[Lemma]{lem:XOR}:
\begin{lemma}\label{lem:xor-new}
Let $\CC_{u\setminus v} = \{C \in \CC(G):(u,u') \in C, (v,v') \not\in C\}$ be the set of cycle covers of $G$ containing $(u,u')$ but not $(v,v')$, and $\CC_{v\setminus u} = \{C \in \CC(G):(v,v') \in C, (u,u') \not\in C\}$. Then there exists a mapping $\psi$ from $\CC(G)$ to subsets of $\CC'(G)$ such that
$\psi(C_1) \cap \psi(C_2) = \emptyset$ for all $C_1,C_2 \in \CC(G)$ and
(1) for any $C \in \CC_{u\setminus v}$, the total weight of $\psi(C)$ is $\sum_{C' \in \psi(C)} \sgn(C')w(C') = 4\sgn(C)w(C)$, (2) for any $C \in \CC_{v\setminus u}$, $\sum_{C' \in \psi(C)} \sgn(C)w(C') = 4\sgn(C)w(C)J_2$, and (3) the remaining cycle covers in $G'$ have total weight $\sum_{C' \not\in \psi(C) \forall C \in \CC_{u\setminus v} \cup \CC_{v\setminus u}}\sgn(C')w(C') = 0$.
\end{lemma}
\begin{proof}
We start with proving (1). Fix any $C \in \CC_{u\setminus v}$. Notice that $\psi(C)$ consists of all $C' \in \CC(G')$ that contain $(u,a)$, $(c,u')$ and all of $C$'s edges except $(u,u')$. Call this set of common edges $E_C$; by the assumption that $w(u,u') = I_2$, $w(E_C) = w(C)$
The set $\psi(C)$ consists of all possible ways of completing $E_C$ to a cycle cover $C'$ of $G'$ by adding edges to $G'$ so that every vertex has indegree and outdegree $1$. Within $E_C$, the only vertices with deficient degree are $a,b,c$ and $d$. Vertices $b$ and $d$ have indegree and outdegree $0$, while $a$ has indegree $1$ and outdegree $0$, and $c$ has indegree $0$ and outdegree $1$. Note that the edges $(u,a)$ and $(c,u')$ must belong to the same cycle in $C'$, and so the edges in $E_C$ form zero or more completed cycles and an incomplete cycle from $c$ to $a$. The number of completed cycles is $c(C) - 1$, where $c(C)$ is the number of cycles in $C$.

We thus need to add three edges matching the vertices $\{a,b,d\}$ to the vertices $\{b,c,d\}$; call these three edges $E_\XOR$, so that $E_C \cup E_\XOR$ forms a cycle cover $C'$. The weight of $C'$ is therefore $w(C') = w(E_C)w(E_\XOR)$. The sign of $C'$ is $(-1)^{n+4 - c(C')}$, where $c(C')$ is the number of cycles in $C'$. We can see that $c(C')$ is the sum of the number of completed cycles in $E_C$ and the number of cycles among $\{a,b,c,d\}$ assuming the existence of an edge from $c$ to $a$. Hence $c(C') = c(C) - 1 + c(E_\XOR \cup \{(c,a)\}$, and so $\sgn(C') = -\sgn(C)(-1)^{4 - c(E_\XOR \cup \{(c,a)\})} = -\sgn(C)\sgn(E_\XOR \cup \{(c,a)\})$.

Thus, $\sum_{C' \in \psi(C)}\sgn(C')w(C') = -\sgn(C)w(E_C)\sum_{C'\in \psi(C)} \sgn(E_\XOR \cup \{(c,a)\})w(E_\XOR) = \\ -\sgn(C)w(E_C)\det(M_{3,1})$. From \lref[Lemma]{lem:dets}, this is $4\sgn(C)w(E_C) = 4\sgn(C)w(C)$, as required.

The proof of (2) proceeds similarly, except that $\psi(C)$ contains $(v,c)$ and $(a,v')$ instead of $(u,a)$ and $(c,v')$. The set of common edges then has an incomplete path from $a$ to $c$. As a result, we end up with $\sum_{C' \in \psi(C)}\sgn(C')w(C') = -\sgn(C)w(E_C)\det(M_{1,3}) = 4\sgn(C)w(C)J_2$.

To prove (3), we observe that a cycle cover in $\CC(G')$ that contains $(u,a)$ and $(c,u')$ but not $(v,c)$ or $(a,v')$ must fall into $\psi(C)$ for some $C \in \CC_{u\setminus v}$; similarly, any cycle cover containing $(v,c)$ and $(a,v')$ but not $(u,a)$ or $c,u')$ must fall into $\psi(C)$ for some $C \in \CC_{v\setminus u}$. These were already accounted for in the proofs of (1) and (2), so we can concentrate only on the leftover cycle covers.  Partition these leftover cycle covers into equivalence classes based on their edge sets excluding edges wholly within $\{a,b,c,d\}$; namely $C'_1 \sim C'_2$ if and only if $C'_1 \setminus \{a,b,c,d\} \times \{a,b,c,d\} = C'_2 \setminus \{a,b,c,d\} \times \{a,b,c,d\}$. For any equivalence class, its cycle covers must either all (a) contain none of these four edges, (b) contain $(u,a)$ and $(a,v')$ only, (c) contain $(v,c)$ and $(a,v')$ only, or (d) contain all four edges.

Up to sign, the total weights of those equivalence classes in (a) contain a factor of $\det(M)$, those in (b) contain a factor $\det(M_{1,1})$, those in (c) contain a factor $\det(M_{3,3})$, and those in (d) contain a factor $\det(M_{13,13})$. From \lref[Lemma]{lem:dets}, all four of these determinants are $0$, and so the total weights of the cycle covers in any equivalence class is $0$, as is therefore the total weight of all the leftover cycle covers.
\end{proof}

With this in hand, we can prove the key result:
\begin{theorem}
Given a 3SAT formula $\varphi$ with $S$ satisfying assignments, let the graph $H_\varphi$ with weights in $M_2(\F)$ be as defined above. Then $\det(H_\varphi) = aI_2 + bJ_2$, where $a + b = 4^{3m}S$.
\end{theorem}
\begin{proof}

The structure of the proof is similar to the sketch given after \lref[Theorem]{thm:valiant}, though with extra care needed for the complications of working with matrices. In the end, each cycle cover of $G$ ends up with weight $4^{3m}I_2$ or $4^{3m}J_2$, giving the result.

Let us start with $H^0_\varphi$, which we know from \lref[Lemma]{lem:new-consistent} has $\det(H^0_\varphi) = SI_2$. In particular, for each satisfying assignment of $\varphi$, there is a consistent cycle cover of $H^0_\varphi$ of weight $I_2$. There exist $3m$ pairs of edges in $H^0_\varphi$ that when replaced by XOR gadgets will convert $H^0_\varphi$ to $H_\varphi$; each of these pairs contains an external edge in a clause gadget and an external edge in a variable gadget referring to the same literal.

Consider what happens when we replace one of the above-mentioned edge pairs with an XOR gadget, forming a new graph $H^1_\varphi$. From \lref[Lemma]{lem:xor-new}, each cycle cover $C$ that is consistent on this edge pair in $H^0_\varphi$ will be mapped to $\psi(C)$, a set of cycle covers in the new graph whose total signed weight will either be $4I_2$ or $4J_2$. Further, since all of these sets $\psi(C)$ are disjoint and all other cycle covers have total signed weight $0$, the total signed weight of all cycle covers in $H^1_\varphi$ is $\sum_{C \in \CC^1_\con(G)}4K_2(C)$, where $\CC^1_\con(G)$ are those cycle covers of $G$ that are consistent on this edge pair, and $K_2(C)$ is either $I_2$ or $J_2$.

Now suppose a second edge pair is replaced with an XOR gadget, resulting in the graph $H^2_\varphi$. Consider a cycle cover $C$ of $H^0_\varphi$ in $\psi(C)$ that is consistent on both the first and second edge pairs. Then each cycle cover of $H^1_\varphi$ in $\psi(C)$ will be mapped to a set of cycle covers $\psi(\psi(C))$ of $H^2_\varphi$, with signed weight that is $4I_2$ or $4J_2$ multiple of its signed weight in $H^1_\varphi$. The set $\psi(\psi(C))$ therefore has total signed weight of either $16I_2$ or $16J_2$, since all of the images of $\psi$ are disjoint. Once again, the total signed weight of all cycle covers in $H^2_\varphi$ is $\sum_{C \in \CC^{1,2}_\con(G)}16K_2(C)$, where $\CC^{1,2}_\con(G)$ is the set of cycle covers of $G$ consistent on both edge pairs.

Carrying this out over all $3m$ edge pairs to reach $H_\varphi$, we see that every consistent cycle cover of $H^0_\varphi$ becomes a disjoint set of cycle covers in $H_\varphi$ of total signed weight $4^{3m}I_2$ or $4^{3m}J_2$, while all other cycle covers in $H_\varphi$ have total weight $0$. The total weight over all original consistent cycle covers is $\sum_{C \in \CC_\con(G)}4^{3m}K_2(C)$. This therefore takes the form given in the theorem.
\end{proof}

This completes the proof of  \lref[Theorem]{thm:2x2}.

\section{Computing the determinant over upper triangular matrix
algebras}\label{sec:upper}

In this section, we consider the problem of computing the determinant over the algebra of upper triangular matrices of dimension $d$. We show that the determinant over these algebras can be computed in time $N^{O(d)}$, where $N$ denotes the size of the input. We will then generalize this theorem to arbitrary algebras to yield \lref[Theorem]{thm_det_algo}.

Given a field $\F$, we denote by $U_d(\F)$ the algebra of upper triangular matrices of
dimension $d$ with entries from the field $\F$.
%
%
%
\begin{theorem}
	\label{thm_det_upper}
	Let $\F$ be a field. There exists a deterministic algorithm, which when given as input an
	$n\times n$ matrix $M$ with entries from $U_d(\F)$, computes the
	determinant of $M$ in time $\poly(N^{d})$, where $N$ is the size of the input.
\end{theorem}

\begin{proof}
	The algorithm is simple. We write out the expression for the
	determinant of $M$ and note that each entry of $\det(M)$ may be
	written as the sum of $n^{O(d)}$ many determinants of matrices with
	entries from the underlying field. Since each of these can be
	computed in time $N^{O(1)}$, we obtain an $N^{O(d)}$-time algorithm
	for our problem.

	Let $M = (m_{i,j})_{i,j}$, where $m_{i,j}\in U_d(\F)$ for each $i,j\in
	[n]$. Given $m\in U_d(\F)$, we use $m(p,q)$ to denote the $(p,q)$th
	entry of $m$. We have
	\[
	\det(M) = \sum_{\sigma\in S_n} \sgn(\sigma)
	m_{1,\sigma(1)}m_{2,\sigma(2)}\cdots m_{n,\sigma(n)}
	\]

	Consider a product of matrices $m = m_1\cdots m_n$ where each
	$m_i\in U_d(\F)$. For $p,q\in [d]$ such that $p\leq
	q$, we may write the $(p,q)$th entry of $m$ as
	\begin{align}
	m(p,q) &= \sum_{k_1,k_2,\ldots,k_{n-1}\in [d]} m_1(p,k_1)m_2(k_1,k_2)\cdots
	m_n(k_{n-1},q)\notag\\
	&= \sum_{p\leq k_1\leq \cdots \leq k_{n-1}\leq q}m_1(p,k_1)m_2(k_1,k_2)\cdots
	m_n(k_{n-1},q)\label{eq_det_u}
	\end{align}
	where the last equality follows since $m_i(k,l) = 0$ unless $k\leq
	l$. Note that the number of terms in the summation in \eqref{eq_det_u} is equal to the number of increasing sequences of
	length $n$ consisting of elements from $[d]$ and is bounded by
	$n^{O(d)}$.

	Fix any $p,q\in [d]$ such that $p\leq q$. By \eqref{eq_det_u}, we may write $\det(M)(p,q)$ as
	\begin{equation}
	\det(M)(p,q) = \sum_{p\leq k_1\leq \cdots \leq k_{n-1}\leq
	q}\sum_{\sigma\in S_n}\sgn(\sigma)\cdot
	m_{1,\sigma(1)}(p,k_1)\cdot m_{2,\sigma(2)}(k_1,k_2)\cdots
	m_{n,\sigma(n)}(k_{n-1},q)\label{eq_det_u_2}
	\end{equation}

	We now note that each of the inner summations may be written as the
	determinant of an appropriate matrix over the underlying field. Fix
	any $\overline{k}= (k_1,\ldots,k_{n-1})$ satisfying $p\leq k_1\leq
	k_2 \leq \cdots \leq k_{n-1}\leq q$. Denote by
	$M_{\overline{k}}$ the matrix $(m_{i,j}(k_{i-1},k_i))_{i,j}$, where
	$k_0$ denotes $p$ and $k_1$ denotes $q$. It follows from \eqref{eq_det_u_2} that $\det(M)(p,q) = \sum_{\overline{k}}
	\det(M_{\overline{k}})$.

	Note that the matrices $M_{\overline{k}}$ are $n\times n$ matrices
	with entries from the underlying field and hence, their determinants
	can be computed in time $N^{O(1)}$. Hence, we can compute
	$\det(M)(p,q)$ --- for each $p,q$ --- in time $n^{O(d)}\cdot
	N^{O(1)} = N^{O(d)}$. The result follows.
\end{proof}

\section{Determinant computation over general algebras}\label{sec:general}

We now consider the problem of computing the determinant of an
$n\times n$ matrix with entries from a general finite-dimensional
algebra $A$ of dimension $D$ over a field $\field$ that is either
finite, or the rationals.  We consider two algorithmic questions: the
first is the problem of computing the determinant over $A$, where $A$
is a \emph{fixed} algebra (and hence of constant dimension) such as
$M_2(\field)$; the second is the case when $A$ is presented to the
algorithm along with the input (in this case, $A$ could have large
dimension). We present our results for the latter case in the
appendix.

In the first case, we prove a strong dichotomy for finite fields of
characteristic $p > 2$. For any fixed algebra $A$, we show, based on
the structure of the algebra, that either the determinant over $A$ is
polynomial-time computable, or computing the determinant over $A$ is
$\mathrm{Mod}_pP$-hard.

We first recall a few basic facts about the structure of finite
dimensional algebras. An algebra is \emph{simple} if it is isomorphic
to a matrix algebra (possibly of dimension $1$) over a field extension
of $\field$. An algebra is said to be \emph{semisimple} if it can be
written as the direct sum of simple algebras. \footnote{This is not
the standard definition of semisimplicity in the case of infinite
fields. However, we will only use it in the case that $\field$ is
finite. See \lref[Appendix]{section_alg_ref}.}

Recall that a \emph{left ideal} in an algebra $A$ is a subalgebra $I$
of $A$ such that for any $x\in I$ and $a\in A$, we have $ax\in I$; a
\emph{right ideal} is defined similarly. An ideal $I$ is said to be
nilpotent if there exists an $m\geq 1$ such that the product of any
$m$ elements from $I$ is $0$. The \emph{radical of $A$} denoted $R(A)$
is defined to be the ideal generated by all the nilpotent left ideals
of $A$. We list some well-known properties of the radical (see
\cite[Chapter IV]{CurtisR1962}): (a) The radical is a left and right ideal in
$A$, (b) The radical is nilpotent: that is, there exists a
$d\in\naturals$ such that the product of any $d$ elements of $R(A)$ is
$0$. The least such $d$ is called the \emph{nilpotency index of the
radical $R(A)$}, and (c) $A/R(A)$ is semisimple.

An algebra $A$ is a \emph{semidirect sum} of subalgebras
$B_1$ and $B_2$ if $A = B_1\oplus B_2$ as a vector space; we
denote this as $A = B_1\oplus' B_2$. The \emph{Wedderburn-Malcev
theorem} (\lref[Theorem]{thm_wedderburn_malcev}) tells us that any
algebra is a semidirect sum of its radical with a subalgebra. We
refer to such a decomposition as a \emph{Wedderburn-Malcev
decomposition}.

We start with the hardness result.

\begin{theorem}
	\label{thm_algebra_hardness}
	Let $A$ denote any fixed algebra over a finite field $\field$ of
	characteristic $p > 2$. If $A/R(A)$ is non-commutative, computing
	the determinant over $A$ is $\mathrm{Mod}_p P$-hard.
\end{theorem}

\begin{proof}
	Consider the problem of computing the determinant over an algebra $A$
	such that $A/R(A)$, the ``semisimple part'' of $A$, is
	non-commutative. Since $A/R(A)$ is semisimple, we know that $A/R(A)
	\cong \bigoplus_i A_i$, where each $A_i$ is a simple algebra, and
	hence isomorphic to a matrix algebra over a field extension of
	$\field$. If each of the $A_i$s is a matrix algebra of dimension $1$
	(that is, each $A_i$ is simply a field extension of $\field$), then
	$A/R(A)$ is commutative. Hence, w.l.o.g., we assume that $A_1$ has
	dimension greater than $1$. Moreover, by the Wedderburn-Malcev
	theorem (see \lref[Theorem]{thm_wedderburn_malcev} in the appendix),
	we know that $A$ contains a subalgebra $B\cong A/R(A)$.  Thus, the
	algebra $A_1$ is isomorphic to a subalgebra of $A$. Thus, \lref[Theorem]{thm:2x2} immediately implies that computing the determinant
	over $A$ is $\mathrm{Mod}_p P$-hard.
\end{proof}

\subsection{The upper bound}

In this section, we show that if $A/R(A)$ is commutative, then the
determinant over $A$ is efficiently computable. However, we present
our result in some generality, which will be useful later. We assume
that the algebra $A$ is presented to the algorithm along with the
input as follows: we are given a (vector space) basis
$\{a_1,\ldots,a_D\}$ for $A$ along with the pairwise products
$a_ia_j$ for every $i,j\in [D]$.  Let $d$ denote the nilpotency index
of $R(A)$.

The Wedderburn-Malcev theorem (see \lref[Theorem]{thm_wedderburn_malcev}) tells us that the algebra $A = B \oplus'
R(A)$, where $B$ is a semisimple subalgebra of $A$ isomorphic to
$A/R(A)$, and hence commutative.

We will use without explicit mention the following result, which was
explicit in the work of Chien and Sinclair \cite{ChienS2007}, and implicit in
that of Mahajan and Vinay \cite{MahajanV1997} (and also many other works):

\begin{theorem}
	\label{thm_comm_alg}
	There is a deterministic algorithm which, when given any commutative
	algebra $A$ of dimension $D$ and an $n\times n$ matrix over $A$ as
	input, computes the determinant of $A$ in time $\poly(n,D)$.
\end{theorem}

We start with two simple lemmas.

\begin{lemma}	
	\label{lemma_nil_index}
	There is a deterministic polynomial-time algorithm which, when given
	an algebra $A$ as input, computes the nilpotency index of $A$.
\end{lemma}

\begin{proof}
	Let $d$ denote the nilpotency index of $A$. It is easy to see that
	$d\leq D$, the dimension of the algebra as a vector space over
	$\field$. The algorithm computes a basis for $R(A)$ (this can be
	done in deterministic polynomial time by \lref[Theorem]{thm_radical_compute}), and then successively computes a basis
	for $R(A)^2, R(A)^3, \ldots, R(A)^D$ and outputs the least $d$ such
	that $R(A)^d = \{0\}$.
\end{proof}

\begin{lemma}
	\label{lemma_id_B}
	Let $A$ be a finite-dimensional algebra with Wedderburn-Malcev
	decomposition $A = B\oplus' R(A)$. Then, $1\in B$.
\end{lemma}

\begin{proof}
	We can write the identity $1$ of $A$ as $1 = b+r$, where $b\in B$
	and $r\in R(A)$. We would like to show that $r=0$. Note that
	$b = b\cdot 1 = b^2 + br$. Since $b^2 \in B$ and $br\in R$, we must
	have $br = 0$. Similarly, $r = 1r = br + r^2$. But $br = 0$ implies
	that $r = r^2$. This implies that $r = r^k$ for any $k \geq 1$. But
	we know that $r$ is nilpotent. Hence, $r = 0$.
\end{proof}

These lemmata and a generalization of \lref[Theorem]{thm_det_upper} yield the following:
\begin{theorem}
	\label{thm_det_algo}
	There exists a deterministic algorithm, which when given as input an
	algebra $A$ of dimension $D$ s.t.  $A/R(A)$ is commutative and an
	$n\times n$ matrix $M$ with entries from $A$, computes the
	determinant of $M$ in time $N^{O(d)}$, where $d$ is the nilpotency
	index of $R(A)$ and $N$ is the size of the input.
\end{theorem}

In particular, when $A$ is a fixed algebra, then $d\leq D = O(1)$, and
hence, \lref[Theorem]{thm_det_algo} gives us a polynomial-time
algorithm. This yields straightaway the sharp dichotomy theorem in
the case of a fixed algebra over finite fields of odd characteristic.

\begin{corollary}
	\label{corollary_dichotomy}
	Let $\field$ be any finite field of odd characteristic and $A$ be
	any fixed algebra over $\field$. Then, if $A/R(A)$ is
	non-commutative, computing the determinant over $A$ is
	$\mathrm{Mod}_pP$-hard. If $A/R(A)$ is commutative, then the
	determinant can be computed in polynomial time.
\end{corollary}

\begin{proof}[Proof of {\lref[Theorem]{thm_det_algo}}]
	The algorithm first computes the Wedderburn-Malcev decomposition $A
	= B \oplus' R(A)$ of the algebra $A$: a result of de Graaf et al.
	(\lref[Theorem]{thm_wedderburn_malcev_compute}) shows that such a
	decomposition may be computed efficiently. By \lref[Lemma]{lemma_nil_index}, we can compute the nilpotency index $d$ of
	the algebra in deterministic polynomial time. We assume that $d\leq
	n$; otherwise, the bruteforce algorithm for the determinant
	has running time $N^{O(d)}$.
	
	For any $i$ and $j$, the $(i,j)$th entry of the input matrix $M$ can
	be written uniquely as $m_{i,j} = b_{i,j}+r_{i,j}$ where $b_{i,j}\in B$
	and $r_{i,j}\in R$; the elements $b_{i,j}$ and $r_{i,j}$ are also
	efficiently computable. Now, note that the determinant of the input
	matrix $M$ can be written as
	$$
   \det(M) = \sum_{\sigma\in S_n} \sgn(\sigma)
		(b_{1,\sigma(1)} + r_{1,\sigma(1)})\cdot (b_{2,\sigma(2)} +
		r_{2,\sigma(2)})\cdots (b_{n,\sigma(n)} + r_{n,\sigma(n)})
		= \sum_{\sigma\in S_n}\sgn(\sigma)\sum_{S\subseteq [n]}
		t(\sigma,S)
	$$
	where $t(\sigma,S)$ is the product, in increasing order of $i$, of
	$r_{i,j}$ for $i\in S$ and $b_{i,j}$ for $i\not\in S$. Note that
	$t(\sigma,S)\in R(A)^{|S|}$ (we use here the fact that $R(A)$ is
	an ideal in $A$) and hence, $t(\sigma,S) = 0$ if $|S|\geq d$. Thus,
	we may only consider $S$ of size strictly less than $d$.
	
	We divide the terms $t(\sigma,S)$ based on the $r_{i,j}$
	that actually appear in $t(\sigma,S)$. Specifically, for each
	$1$-$1$ function $f:S\rightarrow [n]$, let $t(\sigma,S,f)$ denote
	$t(\sigma,S)$ if $\sigma|_S = f$ and $0$ otherwise. We can write the
	determinant $\det(M)$ as
	\begin{align*}
		\det(M) = \sum_{\substack{S\subseteq [n]:\\ |S|<
		d}}\sum_{\substack{f:S\rightarrow
		[n]:\\ f\ 1-1}}\sum_{\sigma\in S_n}\sgn(\sigma) t(\sigma,S,f)
		= \sum_{\substack{S\subseteq [n]:\\ |S|<
		d}}\sum_{\substack{f:S\rightarrow
		[n]:\\ f\ 1-1}}\det(M(S,f))
	\end{align*}
	where the entries $m(S,f)_{i,j}$ of $M(S,f)$ are defined as follows:
	for $i\in S$, $m(S,f)_{i,j} = 0$ if $f(i)\neq j$ and $r_{i,j}$
	otherwise; for $i\not\in S$, $m(S,f)_{i,j} = b_{i,j}$. We show that
	for each $S$ and $f$ as above, $\det(M(S,f))$ can be computed in
	time $N^{O(d)}$, which will prove the theorem,
	since there are only $n^{O(d)}$ of them to compute. For the
	remainder of the proof, we fix some subset $S\subseteq [n]$ of size
	$t < d$ and $f:S\rightarrow [n]$ that is $1$-$1$.

	Note that the matrix $M(S,f)$ is ``almost'' a matrix over the
	commutative subalgebra $B$ of $A$: it contains exactly $d$ entries
	outside $B$, one in each row indexed by an element of $S$. We
	reduce the computation of $\det(M(S,f))$ to the computation of
	the determinant of a similar matrix over a commutative algebra
	closely related to $B$.
	Indeed, let $B^{\otimes (t+1)}$ denote $B\otimes B\otimes\cdots \otimes B$
	($t+1$ times). This is a commutative algebra of dimension
	at most $D^d$. Furthermore, we see that $1^{\otimes (t+1)}$ is the
	identity element of this algebra.  For $i\in [t]\cup \{0\}$, we denote by
	$B_i$ the following subalgebra of $B^{\otimes (t+1)}$: $1^{\otimes
	i}\otimes B\otimes 1^{\otimes (t-i)}$. It can easily be seen
	that each $B_i$ is isomorphic to $B$ by the isomorphism
	$\phi_i:B\rightarrow B_i$ where $\phi_i(b) = 1^{\otimes i}\otimes b\otimes 1^{\otimes (t-i)}$.

	\newcommand{\Pre}{\mathrm{Pre}}
	Given any $i\in [n]$, we denote by $\mathrm{Pre}(i)$ the set
	$\setcond{i'\in S}{i' < i}$. We now construct a new matrix $M'(S,f)$
	with entries from $B^{\otimes (t+1)}$ as follows:
	\[
	m'(S,f)_{i,j} = \left\{
	\begin{array}{lr}
		0 & \text{if $i\in S$ and $f(i)\neq j$,}\\
		1^{\otimes (t+1)} & \text{if $i\in S$ and $f(i)=j$,}\\
		\phi_\ell(m(S,f)_{i,j}) & \text{if $i\not\in S$ and $\ell =
		|\Pre(i)|$}.
	\end{array}\right.
	\]
	In words, to construct $M'(S,f)$, we have replaced each entry in
	$M(S,f)$ that is in $R(A)$ by the identity $1^{\otimes (t+1)}$ and
	each entry $b_{i,j}\in B$ by the corresponding element in $B_\ell$ where
	$\ell = |\Pre(i)|$.

	Since $M'(S,f)$ is a matrix with entries from the commutative
	algebra $B^{\otimes (t+1)}$, its determinant can be computed in time
	$N^{O(d)}$. Say $S = \{i_1,\ldots,i_t\}$ and $f(i_k) = j_k$ for
	$k\in [t]$. Let $\{e_1,\ldots,e_m\}$ be a basis for $B$. Then, we have
	\begin{align*}
		\det(M(S,f)) &= \sum_{\substack{\sigma\in S_n:\\ \sigma|_S = f}}
		\sgn(\sigma)\left(\prod_{i<i_1}b_{i,\sigma(i)}\right)\cdot
		r_{i_1,j_1}\cdot\left(\prod_{i_1<i<i_2}b_{i,\sigma(i)}\right)\cdot
		r_{i_2,j_2}\cdots r_{i_t,j_t}\cdot \left(\prod_{i>i_t}
		b_{i,\sigma(i)}\right)\\
	\end{align*}

	Each product of the form $\prod_{i\in T}b_{i,\sigma(i)}$ that
	appears in the summation above is an element of the commutative
	algebra $B$ and hence can be expanded in the basis of $B$ as follows:
	\begin{align}
		\det(M(S,f)) &= \sum_{\substack{\sigma\in S_n:\\ \sigma|_S = f}}
		\sgn(\sigma)\left(\sum_{k=1}^m \alpha_{k,\sigma}^{(0)}e_k\right)\cdot
		r_{i_1,j_1}\cdot\left(\sum_{k=1}^m \alpha_{k,\sigma}^{(1)}e_k\right)\cdot
		r_{i_2,j_2}\cdots r_{i_t,j_t}\cdot \left(\sum_{k=1}^m
		\alpha_{k,\sigma}^{(t)}e_k\right)\notag\\
		&= \sum_{k_0,\ldots,k_t\in
		[m]}\left(\sum_{\sigma\in S_n: \sigma|_S = f}
		\sgn(\sigma)\alpha_{k_0,\sigma}^{(0)}\cdots\alpha_{k_t,\sigma}^{(t)}\right)e_{k_0}r_{i_1,j_1}e_{k_1}\cdots
		r_{i_t,j_t}e_{k_t}\notag\\
		& = \sum_{\overline{k}\in [m]^{t+1}}c_{\overline{k}}\cdot e_{k_0}r_{i_1,j_1}e_{k_1}\cdots
		r_{i_t,j_t}e_{k_t}\label{eq_msf}
	\end{align}

	\noindent where $\overline{k}$ denotes the tuple $(k_0,\ldots,k_t)$ and
	$c_{\overline{k}}$ denotes $\sum_{\sigma: \sigma|_S =
	f}\alpha_{k_0,\sigma}^{(0)}\cdots\alpha_{k_t,\sigma}^{(t)}$.  Let us
	expand $\det(M'(S,f))$ similarly. We use $e_k^{(\ell)}$ to denote
	$\phi_\ell(e_k)$. We have
	\begin{align*}
		\det(M'(S,f)) &= \sum_{\substack{\sigma\in S_n:\\ \sigma|_S = f}}
		\sgn(\sigma)\left(\prod_{i<i_1}\phi_0(b_{i,\sigma(i)})\right)\cdot
		1^{\otimes (t+1)}\cdot\left(\prod_{i_1<i<i_2}\phi_1(b_{i,\sigma(i)})\right)\cdot
		1^{\otimes (t+1)}\cdots 1^{\otimes (t+1)} \cdot \left(\prod_{i>i_t}
		\phi_t(b_{i,\sigma(i)})\right)\\
		&= \sum_{\substack{\sigma\in S_n:\\ \sigma|_S = f}}
		\sgn(\sigma)\left(\sum_{k=1}^m \alpha_{k,\sigma}^{(0)}e_k^{(0)}\right)\cdot
		1^{\otimes (t+1)}\cdot\left(\sum_{k=1}^m
		\alpha_{k,\sigma}^{(1)}e_k^{(1)}\right)\cdot
		1^{\otimes (t+1)}\cdots 1^{\otimes (t+1)}\cdot \left(\sum_{k=1}^m
		\alpha_{k,\sigma}^{(t)}e_k^{(t)}\right)\\
		&= \sum_{\overline{k}\in [m]^{t+1}}c_{\overline{k}}\cdot
		e_{k_0}^{(0)}e_{k_1}^{(1)}\cdots
		e_{k_t}^{(t)}
		= \sum_{\overline{k}\in [m]^{t+1}}c_{\overline{k}}\cdot
		e_{k_0}\otimes e_{k_1}\otimes \cdots \otimes
		e_{k_t}
	\end{align*}

	Thus, we can simply read off the coefficients
	$c_{\overline{k}}$ from $\det(M'(S,f))$, and using Equation
	\eqref{eq_msf}, we can compute $\det(M(S,f))$. Since $\det(M'(S,f))$
	can be computed in time $N^{O(d)}$, we obtain a
	$N^{O(d)}$-time algorithm to compute $\det(M(S,f))$ and hence for
	$\det(M)$ as well.
\end{proof}


\section{Discussion}\label{sec:disc}

Our results show that the basic Godsil-Gutman approach to
approximating the permanent, as generalized by Chien et al. \cite{ChienRS2003}
runs into many obstacles, since the estimators are not efficiently
computable. In the case of the quaternions, the result of Chien et al.
shows that a suitable modification of the basic estimator still gives
a relatively good approximation to the permanent. Is there such a
modification for matrix algebras?

Our dichotomy theorem in \lref[Section]{sec:general} used crucially
the fact that we worked over a finite field. Over the rationals, for
example, even the structure of semisimple algebras is fairly
complicated, and we don't have an exact characterization of when the
determinant over such an algebra is efficiently computable. Extending
our dichotomy theorem to these algebras is an interesting problem.

\lref[Theorem]{thm_det_algo} shows that even when given the algebra $A$ as
input, the determinant remains efficiently computable as long as
$A/R(A)$ is commutative and $A$ has bounded nilpotence index. How
close is this to being a characterization of algebras over which the
determinant is polynomial-time computable (under reasonable complexity
assumptions) when the algebra is part of the input? More generally,
can one come up with suitable conditions on the radical $R(A)$ under
which computing the determinant over $A$ is hard even when $A/R(A)$ is
commutative?
{\small
\bibliographystyle{prahladhurl}
\bibliography{smalldet}
}

\appendix

\section{Computing the structure of algebras}
\label{section_alg_ref}

An algebra is \emph{simple} if it is isomorphic to a matrix algebra
(possibly of dimension $1$) over a division ring containing $\field$.
Note that if $\field$ is finite, Wedderburn's Little Theorem \cite{Lam1991}
implies that the division ring is a field extension of $\field$ and
hence, a simple algebra is simply a matrix algebra over a field
extension of $\field$. An algebra is said to be \emph{semisimple} if
it can be written as the direct sum of simple algebras.

Friedl and Ronyai \cite{FriedlR1985} first considered the question of
efficiently computing the structural elements of an algebra given as
input in the form of a multiplication table. That is, the algebra $A$
is presented to the algorithm in the form of a basis
$\{a_1,\ldots,a_D\}$ along with a table that lists the pairwise
products $a_ia_j$ for $i,j\in [D]$. They proved the following result:

\begin{theorem}[\cite{FriedlR1985}, Theorem 5.7]
	\label{thm_radical_compute}
	There is a deterministic algorithm which, when given a finite
	dimensional algebra $A$ over $\field$, computes a basis for the
	radical $R(A)$. The algorithm runs in time $\poly(N,\log |\field|)$,
	where $N$ denotes the size of the input.
\end{theorem}

Thus, using the above algorithm, we can obtain in deterministic
polynomial time a description of $A/R(A)$, the ``semisimple'' part of
$A$. Friedl and Ronyai \cite{FriedlR1985} and Ronyai \cite{Ronyai1987} showed
respectively that semisimple algebras can further be decomposed into
simple algebras, and when $\field$ is finite, one can find explicit
isomorphisms from simple algebras to matrix algebras. We state these
two results below.

\begin{theorem}[{\cite[Theorem 7.8]{FriedlR1985}}, {\cite[Theorem 6.2]{Ronyai1987}}]
	\label{thm_semisimple_compute}
	There is a deterministic algorithm which, when given a finite
	dimensional semisimple algebra $A$ over a finite field $\field$,
	computes a decomposition of $A = \bigoplus_i A_i$ into simple
	matrix algebras, and explicit isomorphisms from each $A_i$ to a
	matrix algebra over a field extension of $\field$. The algorithms
	run in time $\poly(N)$, where $N$ is the size of the input.
\end{theorem}

We say that an algebra $A$ is a \emph{semidirect sum} of subalgebras
$B_1$ and $B_2$ if $A = B_1\oplus B_2$ as a vector space. We will
denote this as $A = B_1\oplus' B_2$. The \emph{Wedderburn-Malcev
theorem} tells us that any algebra over a finite field or the
rationals is a semidirect sum of its radical with a subalgebra.

\begin{theorem}[{\cite[Chapter X]{CurtisR1962}}]
	\label{thm_wedderburn_malcev}
  Given any finite dimensional algebra $A$ over a finite field
	or the rationals, there exists a subalgebra $B$ of $A$ such that $A =
	B\oplus' R(A)$.
\end{theorem}

A decomposition of the algebra $A$ as given above is called a
\emph{Wedderburn-Malcev decomposition} of $A$.

\begin{remark}
	Note that by definition, given any Wedderburn-Malcev decomposition
	of $A = B \oplus' R(A)$, the subalgebra $B$ is isomorphic to
	$A/R(A)$, the ``semisimple'' part of $A$.
\end{remark}

The result of de Graaf et al. \cite{deGraafIKR1997} shows that given an algebra
$A$ and the quotient $A/R(A)$, it is possible to obtain a
Wedderburn-Malcev decomposition of $A = B\oplus R(A)$ in deterministic
polynomial time. We state the result below.

\begin{theorem}[{\cite[Theorem 3.1]{deGraafIKR1997}}]
	\label{thm_wedderburn_malcev_compute}
	There is a deterministic polynomial-time algorithm which, when given
	a finite dimensional algebra $A$ over $\field$, computes a
	Wedderburn-Malcev decomposition $A = B\oplus' R(A)$ of the algebra.
\end{theorem}


\section{Computing the determinant over a given algebra}
\label{section_alg_input}

In this section, we assume that the algebra $A$ is presented to the
algorithm in the form of a basis $\{a_1,\ldots,a_D\}$ along with a
table that lists the pairwise products $a_ia_j$ for $i,j\in [D]$. We
would like to efficiently compute the determinant of an input $n\times
n$ matrix with entries from the algebra $A$. We assume that $D \leq
\poly(n)$. (It is easy to see that when $D$ is very large, say
$n!$, even the bruteforce algorithm is efficient. So we assume that
$D$ is small.)

Note that under no constraints on the algebra $A$, this problem is at
least as hard as computing the determinant over $M_2(\field)$ and
hence the hardness results of \lref[Theorem]{thm:2x2} apply. We would like general
conditions on the algebra under which the problem becomes tractable.
When $\field$ is finite, such conditions should ensure that the
semisimple part $A/R(A)$ is commutative (or else $M_2(\field)$ is a
subalgebra of $A$). By the Wedderburn-Malcev theorem, this implies
that $A = B\oplus' R(A)$, where $B$ is a commutative subalgebra of
$A$. We would like general conditions on $R(A)$ under which the
determinant is efficiently computable.

At this point, let us look at an important example that motivates this
work. Consider the algebra $U_d(\field)$ of $d\times d$
upper-triangular matrix algebras over $\field$. Let $D_d(\field)$
denote the (commutative) subalgebra of $d\times d$ diagonal matrices
and let $U'_d(\field)$ be the subalgebra of strictly upper triangular
matrices (i.e., elements of $U_d(\field)$ that contain only zeroes
along the diagonal). It is well-known (see \cite[Section 2.5]{rep-theory2009}, for example) that $R(U_d(\field)) = U'_d(\field)$ and that $U_d(\field) =
D_d(\field) \oplus' U'_d(\field)$ is a Wedderburn-Malcev decomposition
of $U_d(\field)$. Thus, $U_d(\field)/R(U_d(\field))$ is commutative
and hence, this is the kind of algebra we would like to consider. It
can be shown (and we will see below) that when $d$
is constant, the determinant over $U_d(\field)$ can be computed in
polynomial time. On the other hand, when $d$ is $n^{\Omega(1)}$, then
it is known that the problem is hard; this is implicit in an earlier
result of Arvind and Srinivasan \cite{ArvindS2010}, and can also be proved from
\lref[Theorem]{thm:2x2} above.  Therefore, any criterion that
characterizes algebras w.r.t. the tractability of the determinant over
them must explain this difference.

We suspect that this criterion is the nilpotency index of
the radical of the algebra in question. It is easy to see that the
nilpotency index of of $U'_d(\field)$ is $d$. In this case, we can use
\lref[Theorem]{thm_det_algo}, which implies the following:

\begin{corollary}
	\label{cor_det_algo}
	For any constant $d$ and field $\field$ that is either finite or the
	rationals, there is a deterministic polynomial-time algorithm
	running in time $\poly(N^d)$, which when given as input the
	description of an algebra $A$ over $\field$ with nilpotence index
	bounded by $d$ and an $n\times n$ matrix $M$ with entries from $A$,
	computes the determinant of $M$.
\end{corollary}

Thus, the nilpotency index successfully ``explains'' the tractability
of computing the determinant over $U_d(\field)$, when $d = O(1)$. But
does the nilpotency index explain every such instance? We suspect that
the issue of when exactly the determinant becomes intractable is
closely related to the cases when $d$ becomes large, but we are unable
to prove this.

\end{document}